\documentclass[11pt]{article}%
\usepackage{amsfonts}
\usepackage{amsmath}
\usepackage{amssymb}
\usepackage{fullpage}
\usepackage{graphicx}
\usepackage{color}%
\providecommand{\U}[1]{\protect\rule{.1in}{.1in}}

\newtheorem{theorem}{Theorem}

\newtheorem{lemma}[theorem]{Lemma}

\newenvironment{proof}[1][Proof]{\noindent\textbf{#1.} }{\ \rule{0.5em}{0.5em}}

\newcommand{\tr}{\text{Tr}}

\newcommand{\bal}{\begin{align*}}
\newcommand{\eal}{\end{align*}}
\newcommand{\bpm}{\begin{pmatrix}}
\newcommand{\epm}{\end{pmatrix}}

\newcommand{\balign}{\begin{align*}}
\newcommand{\ealign}{\end{align*}}

\def\cH{\mathcal{H}}

\def\cD{\mathcal{D}}
\def\cN{\mathcal{N}}

\def\cM{\mathcal{M}}

\def\tR{\widetilde{R}}
\newcommand{\be}{\begin{equation}}

\newcommand{\ee}{\end{equation}}

\newcommand{\bea}{\begin{eqnarray}}
\newcommand{\eea}{\end{eqnarray}}

\DeclareMathOperator{\id}{id}
\def\reff#1{(\ref{#1})}
\DeclareRobustCommand\openone{\leavevmode\hbox{\small1\normalsize\kern-.33em1}}
\begin{document}

\title{\textbf{Quantum-to-classical rate distortion coding}}
\author{Nilanjana Datta\\\textit{Statistical Laboratory, University of Cambridge},\\
\textit{Wilberforce Road, Cambridge CB3 0WB, United Kingdom}
\and Min-Hsiu Hsieh\\\textit{Centre for Quantum Computation \& Intelligent Systems},\\
\textit{Faculty of Engineering and Information Technology,}\\
\textit{University of Technology, Sydney, P.O.~Box 123, Broadway NSW 2007, Australia}
\and Mark M. Wilde\\\textit{School of Computer Science, McGill University,}\\
\textit{Montreal, Quebec H3A 2A7, Canada}
\and Andreas Winter\\
\textit{ICREA -- Instituci\'{o} Catalana de Recerca i Estudis Avan\c{c}ats,}\\
\textit{Pg. Lluis Companys 23, ES-08010 Barcelona, Spain, and}\\
\textit{F\'{\i}sica Te\`{o}rica: Informaci\'{o} i Fenomens Qu\`{a}ntics,}\\
\textit{Universitat Aut\`{o}noma de Barcelona, ES-08193 Bellaterra (Barcelona), Spain}\\[1mm]
\textit{Department of Mathematics, University of Bristol, Bristol BS8 1TW, United Kingdom}\\[1mm]
\textit{Centre for Quantum Technologies, National University of Singapore}}
\maketitle

\begin{abstract}
We establish a theory of quantum-to-classical rate distortion coding.
In this setting, a sender Alice has many copies of a quantum information source.
Her goal is to transmit classical information about the source, obtained by
performing a measurement on it, to a receiver Bob, up to some
specified level of distortion. We derive a single-letter formula
for the minimum rate of classical communication needed for this
task. We also evaluate this rate in the case in which Bob has some
quantum side information about the source. Our results imply that,
in general, Alice's best strategy is a non-classical one, in which
she performs a collective measurement on successive outputs of the source.
\end{abstract}

\section{Introduction}

A fundamental task in quantum information theory is the reliable compression
of information emitted by a quantum information source, to enable efficient
storage of the data. Schumacher \cite{Schumacher:1995dg} proved that, for a
memoryless source, the optimal rate of \emph{{lossless}} data compression (in
which the original data is recovered perfectly in the limit of asymptotically
many copies of the source) is given by the von Neumann entropy of the source.
The corresponding rate for a classical source is given by its Shannon entropy
\cite{Shannon:1948wk}.

In realistic applications it may be possible, however, to tolerate imperfect
recovery of the signals, and hence allow for a bounded distortion of the
original information. In fact, this may even be necessitated by the lack of
sufficient storage. These considerations have led to the development of rate
distortion theory \cite{B71}, which is the theory of lossy data compression.
The fundamental results of classical rate distortion theory are attributed to
Shannon \cite{Shannon:tf} and date back to 1948. Its quantum
counterpart was introduced by Barnum \cite{B00} and developed further in
Refs.~\cite{Devetak:2002it,CW08}. Recently, Datta \textit{et al.} identified a
regularized expression for the quantum rate distortion function as well as a
single-letter expression for the entanglement-assisted quantum rate distortion
function~\cite{DHW11}.

In this paper, we consider the situation in which a party (say, Alice) obtains
many copies of a quantum information source described by a quantum state, and
she already has a description of the source in terms of its density operator.
She is only allowed to perform measurements on the source. Her aim is to
suitably compress the classical data resulting from her measurements and send
it to another party (say, Bob) such that, upon decompression, the data
recovered by Bob has a fixed level of distortion from the quantum source
(specified by a suitable distortion observable). Alice is allowed to perform
any measurement that she wishes on the source states to produce a classical
sequence, with the requirement that the average symbol-wise distortion of this
sequence be no larger than some prescribed amount. Analogous to previous
terminology used in quantum information theory, we refer to this as
\emph{{quantum-to-classical}} rate distortion theory, since it deals with an
analysis of the trade-off between the optimal rate of compression of the data
obtained by measurements on the \emph{{quantum source}}, and the allowed
distortion on the recovered \emph{{classical data}}. This trade-off is
quantified by the quantum-to-classical rate-distortion function.

Another way of emphasizing the relevance of quantum-to-classical rate
distortion theory is by adopting the perspective that all classical data
arises from a measurement of a quantum state. This is especially important in
cases where the source is truly non-classical, such as an atomic decaying
process or a highly attenuated laser.\footnote{We note that a similar perspective was used
to justify the development of quantum-to-classical randomness extractors
\cite{BFW11}.} In particular, we can imagine that a memoryless classical
source arises from an appropriate measurement on the states emitted by a
quantum source, and the resulting classical data is some description or
characterization of the original quantum source. Thus, this perspective
necessitates a revision of Shannon's rate-distortion theory \cite{Shannon:tf}
by allowing for an arbitrary measurement to be performed on the original
quantum source. A naive approach to this setting would be to measure each
individual output of the quantum source, treat the resulting classical data as
information emitted by a classical source, and then apply Shannon's
rate-distortion theory to the latter.

Here, we instead allow for collective measurements on the outputs of the
source, and our approach is to apply a derandomized measurement compression
protocol to achieve this task \cite{Winter01a}. We find a single-letter
formula for the quantum-to-classical rate distortion function, expressed as a
minimization of the quantum mutual information over all maps that meet the distortion
constraint. Our result implies that, in general, a quantum strategy is needed
to achieve optimal compression rates and that Shannon's rate-distortion theory
is insufficient in this setting. This result is analogous to the fact that
collective measurements are needed in general in the well-known
Holevo-Schumacher-Westmoreland theorem \cite{Hol98,PhysRevA.56.131} regarding
classical communication over quantum channels (see Ref.~\cite{GGLMSY04} for an
explicit example of a channel for which collective measurements outperform
classical strategies).

In the classical setting, the optimal rate of data compression can be reduced
if the decoder (Bob) has some side information at his disposal. The first
discovery in this direction is due to Slepian and Wolf \cite{SW73}, who showed
that the optimal lossless compression rate is given by the entropy of the
source conditioned on the side information. Wyner and Ziv extended these
results to the case of lossy classical data compression with classical side
information \cite{WZ76}. For the quantum setting, one might imagine that
quantum side information is available at the decoder. In Ref.~\cite{DW03},
Devetak and Winter proved that if Bob has quantum side information at his
disposal, then the optimal lossless compression rate for a classical
information source is reduced from the Shannon entropy of the source by the
Holevo information between the source and the quantum side information. The
case of lossy classical data compression with quantum side information, which
is a quantum generalization of the Wyner-Ziv problem, was studied by Luo and
Devetak \cite{LD09}.

We also study the effect of quantum side information on the above-mentioned
quantum-to-classical rate distortion function. In particular, we consider the
case in which some quantum side information about the original quantum source
is available to Bob. He is allowed to use this information to recover the
classical data obtained from Alice's measurements on the source states. We
also let Alice and Bob share common randomness. In this case, we find a
single-letter formula for the corresponding quantum-to-classical rate
distortion function. One of our assumptions in this setting is that the
process of compression and decompression only causes a negligible disturbance
to the quantum side information. This assumption can be justified by the
possibility of Bob wanting to use the quantum side information in some future
protocol. Our result improves upon the aforementioned work of Luo and Devetak
\cite{LD09} in the sense that we find a matching single-letter converse for
this setting. The achievability part of the proof of this theorem exploits
measurement compression with quantum side information \cite{WHBH12}.

The paper is organized as follows. We summarize some necessary definitions and
prerequisites in Section~\ref{sec-def}, and in Section~\ref{sec-obsvble}, we
review the concept of a distortion observable (originally introduced in
Refs.~\cite{WA01,CW08}). In Section~\ref{sec-qcrd}, we introduce the task of
quantum-to-classical rate distortion coding, define a suitable distortion
observable, and derive an expression for the quantum-to-classical rate
distortion function. In Section~\ref{sec-qsi}, we study quantum-to-classical
rate distortion in the presence of quantum side information and common
randomness.
The main results of this paper are given by Theorem~\ref{thm:qc-rd} of
Section~\ref{sec-qcrd} and Theorems~\ref{thm:qc-qsi} and \ref{thm:qc-cr-qsi}
of Section~\ref{sec-qsi}.

\section{Notations and definitions}

\label{sec-def}Let $\mathcal{B}(\mathcal{H})$ denote the algebra of linear
operators acting on a finite-dimensional Hilbert space $\mathcal{H}$ and let
$\mathcal{D}(\mathcal{H})$ denote the set of positive operators of unit trace
(states) acting on $\mathcal{H}$. For any given pure state $|\psi\rangle
\in\mathcal{H}$ we denote the projector $|\psi\rangle\langle\psi|$ simply as
$\psi$. The trace distance between two operators $A$ and $B$ is given by
$\left\Vert {A-B}\right\Vert _{1}\equiv\text{Tr}|A-B|$, where $|C|\equiv
\sqrt{C^{\dag}C}$. Throughout this paper we restrict our considerations to
finite-dimensional Hilbert spaces, and we take the logarithm to base $2$. In
the following we denote a completely positive trace-preserving (CPTP) map
$\mathcal{N}:\mathcal{B}(\mathcal{H}_{A})\rightarrow\mathcal{B}(\mathcal{H}%
_{B})$ simply as $\mathcal{N}^{A\rightarrow B}$. Similarly we denote an
isometry $U:\mathcal{B}(\mathcal{H}_{A})\rightarrow\mathcal{B}(\mathcal{H}%
_{B}\otimes\mathcal{H}_{E})$ simply as $U^{A\rightarrow BE}$. The identity map
on states in $\mathcal{D}(\mathcal{H}_{A})$ is denoted as $\mathrm{{id}}_{A}$.

The von Neumann entropy of a state $\rho\in\mathcal{D}(\mathcal{H}_{A})$ is
defined as $H(\rho)\equiv-\text{Tr}\{\rho\log\rho\}$. In the following we use
$H(A|B)_{\rho}$ and $I(A;B)_{\rho}$ to respectively denote the conditional
quantum entropy and the quantum mutual information of a bipartite state
$\rho_{AB}$, and $I(A;C|B)_{\sigma}$ to denote the conditional quantum mutual
information for a tripartite state $\sigma_{ABC}$ (see, e.g.,
Refs.~\cite{book2000mikeandike,W11}). We also employ the following properties
of the quantum mutual information:

\begin{lemma}
[Quantum data processing inequality \cite{SN96, W11}]\label{dataproc} If
$\omega_{AB^{\prime}} = (\mathrm{{id}}_{A} \otimes\mathcal{N}^{B\to B^{\prime
}}) \sigma_{AB}$, where $\mathcal{N}^{B\to B^{\prime}}$ is a CPTP map, then
\begin{equation}
I(A;B)_{\sigma}\ge I(A;B^{\prime})_{\omega}.
\end{equation}

\end{lemma}

\begin{lemma}
[Superadditivity of the quantum mutual information \cite{DHW11}]\label{super}
The mutual information is superadditive in the sense that, for any CPTP map
$\mathcal{N}^{A_{1}A_{2} \to B_{1}B_{2}}$,%
\[
I\left(  R_{1}R_{2};B_{1}B_{2}\right)  _{\sigma}\geq I\left(  R_{1}%
;B_{1}\right)  _{\sigma}+I\left(  R_{2};B_{2}\right)  _{\sigma},
\]
where
\[
\sigma_{R_{1}R_{2}B_{1}B_{2}} = \mathcal{N}^{A_{1}A_{2} \to B_{1}B_{2}}
\left(  \phi_{R_{1}A_{1}} \otimes\varphi_{R_{2}A_{2}} \right)  ,
\]
and $\phi_{R_{1}A_{1}}$ and $\varphi_{R_{2}A_{2}}$ are pure bipartite states.
\end{lemma}

In proving our first theorem (Theorem~\ref{thm:qc-rd} of Section
\ref{sec-qcrd}) we make use of the \textquotedblleft measurement
compression\textquotedblright\ theorem (Theorem~2 of Ref.~\cite{Winter01a}).
The latter specifies an optimal two-dimensional rate region characterizing the
resources (namely, common randomness and classical communication) needed for
an asymptotically faithful simulation of a measurement on a quantum state. For
an exact statement of the theorem, see Refs.~\cite{Winter01a,WHBH12}. Here we
give a brief description of its content. Let $\psi_{RA}^{\rho}$ denote the
purification of a quantum state $\rho\in\mathcal{D}(\mathcal{H}_{A})$,
multiple copies of which are in Alice's possession. Suppose Alice does a
measurement, given by a positive operator-valued measure (POVM)
$\Lambda\equiv\{\Lambda_{x}\}$, on each of the
systems in her possession. In the ideal measurement compression protocol, the
state of the classical registers containing Alice's measurement outcomes and
the purifying reference systems $R$ is equivalent to many copies of the
following state:
\begin{equation}
\sigma_{XR}\equiv\sum_{x}|x\rangle\langle x|_{X}\otimes\text{Tr}_{A}\left\{
(I_{R}\otimes\Lambda_{x})\psi_{RA}^{\rho}\right\}  . \label{sig}%
\end{equation}

The measurement compression theorem asserts that if Alice and Bob share
$nH(X|R)_{\sigma}$ bits of common randomness, then it is possible for them to
simulate the measurement $\Lambda^{\otimes n}$ on the state $\rho^{\otimes n}$
with approximately $n I(X;R)_{\sigma}$ bits of classical communication, for
$n$ large enough. The simulation becomes faithful in the limit $n \to\infty$,
in the sense that a verifying party who possesses the classical registers and
the reference systems cannot distinguish between the output of the simulation
and the ideal protocol. If no common randomness is present and Alice is
required to obtain the outcomes of the measurement in addition to Bob, then
the classical communication needed is equal to the Shannon entropy
$H(X)_{\sigma}$. In the above, $H(X|R)_{\sigma}$ and $I(X;R)_{\sigma}$
respectively denote the conditional entropy and the mutual information of the
state $\sigma_{XR}$ defined above. For a more detailed statement of the
theorem, see the proof of Theorem~\ref{thm:qc-rd} in Section \ref{sec-qcrd}.

\section{Distortion observables}

\label{sec-obsvble}As discussed in the Introduction, in rate distortion theory
one allows the data which is recovered after the compression-decompression
scheme to be distorted by some finite amount from the original data. There are
various possible choices of the distortion measure, depending on the nature of
the application. For example, in classical rate distortion theory, the Hamming
distance and the mean squared error are natural choices of the distortion
measure \cite{B71,book1991cover}. In quantum rate distortion theory, the
distortion measure is usually defined in terms of the entanglement fidelity
(see, e.g., Refs.~\cite{B00, DHW11} and references therein). However, since
the distortion is a physical quantity, it is natural to associate with it an
observable in the quantum setting (as discussed in Section~II of
Ref.~\cite{CW08} and in unpublished work \cite{WA01}). This is reviewed below.

In the classical setting, let $x\in\mathcal{X}$ denote the letters of a source
alphabet and let $y\in\mathcal{Y}$ denote the letters of a reconstruction
alphabet. Then to determine the distortion between an input and output letter,
one defines a non-negative cost function $d(x,y)$ (e.g., the Hamming distance
or the squared error), and the average distortion is then given by
\begin{equation}
\sum_{x}\sum_{y}p(x)q(y|x)d(x,y), \label{classical}%
\end{equation}
where $q(y|x)$ is the conditional probability of getting the letter $y$ after
reconstruction when the source letter is $x$, and $p(x)$ is the probability of
source letter $x$.

In the quantum case, one defines a \emph{{distortion observable}} $\Delta$
\cite{CW08}. For example, suppose that $\Delta$ is given by
\begin{equation}
\Delta=\sum_{x}\sum_{y}d(x,y)|x\rangle\langle x|\otimes|y\rangle\langle y| ,
\label{del_class}%
\end{equation}
where $|x\rangle$ are the Schmidt vectors of the following purification of the
source state $\rho$:
\begin{equation}
|\psi_{RA}^{\rho}\rangle=\sum_{x}\sqrt{\lambda_{x}}|x\rangle_{R}|x\rangle_{A},
\end{equation}
so that $\rho=\text{Tr}_{R}\{\psi_{RA}^{\rho}\}$.

Then we recover the expression (\ref{classical}) for the average distortion in
the classical case as follows. Let $\Phi:\mathcal{B}(\mathcal{H}_{A}%
)\mapsto\mathcal{B}(\mathcal{H}_{B})$ denote a map on the source state. Then
the average distortion is given by%
\begin{align}
&  \text{Tr}\left\{  \Delta\bigl((\id\otimes\Phi)(\psi_{RA}^{\rho
})\bigr)\right\} \nonumber\\
&  =\text{Tr}\left\{  \Bigl(\sum_{x,y}d(x,y)|x\rangle\langle x|_{R}%
\otimes|y\rangle\langle y|_{B}\Bigr)\Bigl(\sum_{x^{\prime},y^{\prime}}%
\sqrt{\lambda_{x^{\prime}}}\sqrt{\lambda_{y^{\prime}}}|x^{\prime}%
\rangle\langle y^{\prime}|_{R}\otimes\Phi(|x^{\prime}\rangle\langle y^{\prime
}|_{A})\Bigr)\right\} \nonumber\\
&  =\sum_{x,y}d(x,y)\lambda_{x}\langle y|\,\Phi(|x\rangle\langle
x|_{A})|y\rangle.
\end{align}
Let us define $q(y|x)\equiv\langle y|\Phi(|x\rangle\langle x|_{A})|y\rangle$
since it can be interpreted as the conditional probability of the map $\Phi$
yielding the letter $y$, given that the source letter was $x$. Then setting
$p(x)=\lambda_{x}$, (since $\lambda_{x}$, being an eigenvalue of $\rho$, is a
probability), we recover the expression for the classical average distortion
as in (\ref{classical}).

\section{Quantum-to-classical rate-distortion coding}

\label{sec-qcrd}

Consider a memoryless quantum information source $\{\rho,\mathcal{H}_{A}\}$.
In quantum-to-classical (q-c) rate distortion, Alice starts with $n$ copies
$\rho^{\otimes n}$ of the source state and performs a POVM $\Lambda
^{(n)}=\{\Lambda_{x^{n}}\}$ on it, with the POVM elements $\Lambda_{x^{n}}%
\in\mathcal{B}(\mathcal{H}_{A}^{\otimes n})$ being indexed by classical
sequences $x^{n}\in\mathcal{X}^{n}$ ($\mathcal{X}$ being a finite alphabet),
which correspond to the different possible outcomes of the measurement. It is
convenient to define a measurement map $\mathcal{M}_{\Lambda^{(n)}}$
corresponding to the POVM $\Lambda^{(n)}$ as follows: For any $\sigma_{n}%
\in\mathcal{D}(\mathcal{H}_{A}^{\otimes n})$,
\begin{equation}
\mathcal{M}_{\Lambda^{(n)}}(\sigma_{n})\equiv\sum_{x^{n}\in\mathcal{X}^{n}%
}\text{Tr}\left\{  \Lambda_{x^{n}}\sigma_{n}\right\}  |x^{n}\rangle\langle
x^{n}|. \label{meas_op}%
\end{equation}
The above specifies that with probability $\text{Tr}\left\{  \Lambda_{x^{n}%
}\sigma_{n}\right\}  $ the outcome of the POVM $\Lambda^{(n)}$ on the state
$\sigma_{n}$ is given by the classical sequence $x^{n}$.
Figure~\ref{fig:qc-rd}\ depicts the most general protocol for
quantum-to-classical rate-distortion coding.%
\begin{figure}
[ptb]
\begin{center}
\includegraphics[
natheight=3.306200in,
natwidth=6.719600in,
height=1.7564in,
width=3.5405in
]%
{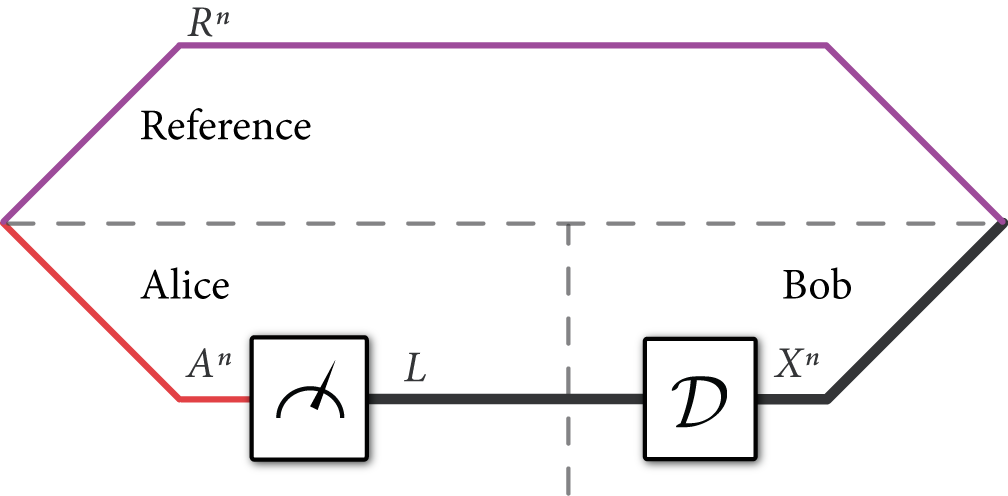}%
\caption{The most general protocol for quantum-to-classical rate-distortion
coding. Alice has many copies of the quantum information source, on which she
performs a collective measurement with classical output $L$. She sends the
variable $L$ over noiseless classical bit channels to Bob. Bob then performs a
classical decoding map on $L$ that outputs the classical sequence $X^{n}$. The
average deviation of this sequence from the quantum source, according to some
distortion observable, provides a measure of the distortion caused by this
protocol.}%
\label{fig:qc-rd}%
\end{center}
\end{figure}

If $\psi_{R^{n}A^{n}}^{\rho}$ denotes a purification of $\rho^{\otimes n}$,
then the following bipartite state characterizes both the classical outcome of
the POVM $\Lambda^{(n)}$ on $\rho^{\otimes n}$ and the post-measurement state
of the purifying reference system:
\begin{align}
\sigma_{R^{n}X^{n}}  &  \equiv\left(  \id_{R^{n}}\otimes\mathcal{M}%
_{\Lambda^{(n)}}\right)  (\psi_{R^{n}A^{n}}^{\rho})\nonumber\\
&  =\sum_{x^{n}}\text{Tr}_{A^{n}}\left\{  \bigl(I_{R^{n}}\otimes\Lambda
_{x^{n}}\bigr)\psi_{R^{n}A^{n}}^{\rho}\right\}  \otimes|x^{n}\rangle\langle
x^{n}|_{X^{n}}.%
\end{align}

We define the q-c distortion measure for a state $\rho\in\mathcal{D}%
(\mathcal{H}_{A})$ with purification $|\psi_{RA}^{\rho}\rangle$ and a POVM
$\Lambda=\{\Lambda_{x}\}$ as
\begin{equation}
d(\rho,\mathcal{M}_{\Lambda})\equiv\text{Tr}\left(  \Delta(\id\otimes
\mathcal{M}_{\Lambda})\left(  \psi_{RA}^{\rho}\right)  \right)  ,
\label{measure}%
\end{equation}
where $\mathcal{M}_{\Lambda}$ is the measurement map corresponding to
$\Lambda$, and $\Delta$ is a q-c distortion observable given by
\begin{equation}
\Delta\equiv\Delta_{RX}\equiv\sum_{x}\Delta_{x}\otimes|x\rangle\langle x|,
\label{qc}%
\end{equation}
with $\Delta_{x}\geq0$.

A q-c rate distortion code of rate $R$ is given by a POVM $\Lambda^{(n)}$ with
$\left\lfloor 2^{nR}\right\rfloor $ outcomes, i.e., $\Lambda^{(n)}
=\{\Lambda_{x^{n}}\}$ with
\begin{equation}
\#\{ x^{n} \in\mathcal{X}^{n} \, : \, \Lambda_{x^{n}} \ne0\} = \left\lfloor
2^{nR}\right\rfloor .
\end{equation}

To define the average distortion resulting from this POVM, we consider a
symbol-wise q-c distortion observable
\begin{equation}
\Delta^{(n)}\equiv\frac{1}{n}\sum_{i=1}^{n}\Delta_{R_{i}X_{i}}\otimes
I_{RX}^{\otimes\lbrack n]\backslash i}, \label{symb}%
\end{equation}
where each operator $\Delta_{R_{i}X_{i}}$ is of the form (\ref{qc}) and
$I_{RX}^{\otimes\lbrack n]\backslash i}$ denotes the identity operator acting
on all but the $i^{th}$ member of the tensor-product of Hilbert spaces
$(\mathcal{H}_{R} \otimes\mathcal{H}_{X})^{\otimes n}$. The average distortion
is then defined as
\begin{align}
{\overline{d}}(\rho,\mathcal{M}_{\Lambda^{(n)}})  &  \equiv\text{Tr}\left(
\Delta^{(n)}(\id_{R^{n}}\otimes\mathcal{M}_{\Lambda^{(n)}})\psi_{R^{n}A^{n}%
}^{\rho}\right) \nonumber\label{avgdist}\\
&  =\frac{1}{n}\sum_{i=1}^{n}\text{Tr}\left(  \Delta_{R_{i}X_{i}}%
\,\sigma_{R_{i}X_{i}}\right)  ,\nonumber
\end{align}
where $\sigma_{R_{i}X_{i}}=\text{Tr}_{\neq i}\sigma_{R^{n}X^{n}}$, with
$\sigma_{R^{n}X^{n}}\equiv(\id_{R^{n}}\otimes\mathcal{M}_{\Lambda^{(n)}}%
)\psi_{R^{n}A^{n}}^{\rho}$. \smallskip

For any $R,D\geq0$, the pair $(R,D)$ is said to be an \emph{{achievable}} q-c
rate distortion pair if there exists a sequence of POVMs $\{\Lambda
^{(n)}\}_{n\ge1}$ of rate $R$ such that
\begin{equation}
\lim_{n\rightarrow\infty}{\overline{d}}(\rho,\mathcal{M}_{\Lambda^{(n)}})\leq
D. \label{eq:rate-dist-criterion}%
\end{equation}
The q-c rate distortion function is then defined as
\begin{equation}
R^{qc}(D)\equiv\inf\{R\,:\,(R,D)\,\,{\hbox{achievable}}\}. \label{def1}%
\end{equation}
\smallskip

\noindent The following theorem provides a single-letter expression for
$R^{qc}(D)$.

\begin{theorem}
\label{thm:qc-rd}For a memoryless quantum information source $\{\rho
,\mathcal{H}_{A}\}$, a quantum-to-classical distortion observable $\Delta_{RX}$,
and any given distortion $D\geq0$, the
quantum-to-classical rate distortion function is given by%
\begin{equation}
R^{qc}\left(  D\right)  =\min_{%
\genfrac{}{}{0pt}{}{\mathrm{{POVM}}\,\Lambda\equiv\{\Lambda_{x}\}}{{d(\rho
,}\mathcal{M}_{{\Lambda}}{)\leq D}}%
}I(X;R)_{\sigma} \label{up1}%
\end{equation}
where $d(\rho,\mathcal{M}_{\Lambda})$ is defined through (\ref{measure}%
)-(\ref{qc}) and
\begin{equation}
\sigma_{RX}\equiv(\id_{R}\otimes\mathcal{M}_{\Lambda})(\psi_{RA}^{\rho}%
)=\sum_{x}{\emph{Tr}}_{A}\left\{  (\id_{R}\otimes\Lambda_{x})\psi_{RA}^{\rho
}\right\}  \otimes|x\rangle\langle x|_{X}. \label{eq:cq-state-for-qc}%
\end{equation}

\end{theorem}

\begin{proof}
We first give the proof of achievability, which follows directly from the
measurement compression theorem \cite{Winter01a} (summarized briefly in the
previous section). Our approach is similar to one used before \cite{W02}: exploit a channel simulation protocol
and derandomize the common randomness consumed by this protocol. So, fix the POVM $\Lambda=\{\Lambda_{x}\}$ that minimizes the
RHS of (\ref{up1}). Thus we have
\begin{equation}
{d}(\rho,\mathcal{M}_{\Lambda})=\text{Tr}\left[  \Delta(\id_{R}\otimes
\mathcal{M}_{\Lambda})\left(  \psi_{RA}^{\rho}\right)  \right]  \leq D.
\label{key}%
\end{equation}
In Ref.~\cite{Winter01a}, it was proved that there exists a finite set of
POVMs $\{\Lambda^{(m)}=\{\Lambda_{x^{n}}^{(m)}\}_{x^{n}\in\mathcal{X}^{n}%
}\,:m=1,2,\ldots,M\}$, each having at most $L$ outcomes, i.e., $\#\{x^{n}%
\in\mathcal{X}^{n}\,:\,\Lambda_{x^{n}}^{(m)}\neq0\}\leq L$, with%
\begin{equation}
L=2^{nI(R;X)_{\sigma}+O(\sqrt{n})}, \label{emm}%
\end{equation}
such that for any ${\varepsilon}>0$ and $n$ large enough, the POVM
$\tilde{\Lambda}^{(n)}=\{\tilde{\Lambda}_{x^{n}}\}$ defined as
\begin{equation}
\tilde{\Lambda}^{(n)}\equiv\frac{1}{M}\sum_{m=1}^{M}\Lambda^{(m)},
\end{equation}
satisfies the following condition:%
\begin{equation}
||\left(  \id_{R^{n}}\otimes\mathcal{M}_{\tilde{\Lambda}^{(n)}}\right)
\psi_{R^{n}A^{n}}^{\rho}-\left(  \id_{R^{n}}\otimes\mathcal{M}_{\Lambda
^{\otimes n)}}\right)  \psi_{R^{n}A^{n}}^{\rho}||_{1}\leq{\varepsilon,}
\label{ext}%
\end{equation}
where for any sequence $x^{n}=x_{1}\ldots x_{n}\in\mathcal{X}^{n}$ we have
$\Lambda_{x^{n}}=\Lambda_{x_{1}}\otimes\ldots \otimes \Lambda_{x_{n}}$.
Further, due to our choice (\ref{symb}) of a \emph{{symbol-wise}} q-c
distortion observable $\Delta^{(n)}$, we have that%
\begin{equation}
\overline{d}(\rho,\mathcal{M}_{\Lambda^{\otimes n}})=\text{Tr}[\Delta
^{(n)}(\id_{R^{n}}\otimes\mathcal{M}_{\Lambda^{\otimes n}})\psi_{R^{n}A^{n}%
}^{\rho}]=d(\rho,\mathcal{M}_{\Lambda}). \label{one2}%
\end{equation}
From (\ref{ext}), we know that the protocol for simulating the tensor-product
measurement has measurement encodings $\{\Lambda_{l}^{\left(  m\right)  }\}$.
Let $\mathcal{D}^{\left(  m\right)  }\left(  l\right)  $ denote the
corresponding classical decodings which construct the sequences $x^{n}$ from
the values of $l$ and $m$, where $l$ is the measurement outcome and $m$ is the
common randomness. Then
\[
\left\Vert (\id_{R^{n}} \otimes\mathcal{M}_{\Lambda^{\otimes n}})\left(
\psi_{R^{n}A^{n}}^{\rho} \right)  -\frac{1}{M} \sum_{m,l}\text{Tr}_{A^{n}%
}\left\{  (\id_{R^{n}} \otimes\Lambda_{l}^{\left(  m\right)  })\psi
_{R^{n}A^{n}}^{\rho} \right\}  \otimes\left\vert \mathcal{D}^{\left(
m\right)  }\left(  l\right)  \right\rangle \left\langle \mathcal{D}^{\left(
m\right)  }\left(  l\right)  \right\vert \right\Vert _{1}\leq\varepsilon.
\]
Then using (\ref{key}) we obtain a bound on the average distortion resulting
from the action of the POVM $\tilde{\Lambda}^{(n)}$ on the source state
$\rho^{\otimes n}$ as follows:%
\begin{align*}
&  {\overline{d}}(\rho,{M}_{\tilde{\Lambda}^{(n)}})\\
&  =\text{Tr}\left\{  \Delta^{\left(  n\right)  }\frac{1}{\left\vert
{M}\right\vert }\sum_{m,l}\text{Tr}_{A^{n}}\left\{  (\id_{R^{n}}
\otimes\Lambda_{l}^{\left(  m\right)  })\psi_{R^{n}A^{n}}^{\rho} \right\}
\otimes\left\vert \mathcal{D}^{\left(  m\right)  }\left(  l\right)
\right\rangle \left\langle \mathcal{D}^{\left(  m\right)  }\left(  l\right)
\right\vert \right\} \\
&  =\frac{1}{\left\vert {M}\right\vert }\sum_{m}\text{Tr}\left\{
\Delta^{\left(  n\right)  }\sum_{l}\text{Tr}_{A^{n}}\left\{  (\id_{R^{n}}
\otimes\Lambda_{l}^{\left(  m\right)  })\psi_{R^{n}A^{n}}^{\rho} \right\}
\otimes\left\vert \mathcal{D}^{\left(  m\right)  }\left(  l\right)
\right\rangle \left\langle \mathcal{D}^{\left(  m\right)  }\left(  l\right)
\right\vert \right\} \\
&  =\frac{1}{\left\vert \mathcal{M}\right\vert }\sum_{m}d_{\max}%
\text{Tr}\left\{  \frac{\Delta^{\left(  n\right)  }}{d_{\max}}\sum
_{l}\text{Tr}_{A^{n}}\left\{  (\id_{R^{n}} \otimes\Lambda_{l}^{\left(
m\right)  })\psi_{R^{n}A^{n}}^{\rho} \right\}  \otimes\left\vert
\mathcal{D}^{\left(  m\right)  }\left(  l\right)  \right\rangle \left\langle
\mathcal{D}^{\left(  m\right)  }\left(  l\right)  \right\vert \right\} \\
&  \leq d_{\max}\text{Tr}\left\{  \frac{\Delta^{\left(  n\right)  }}{d_{\max}%
}(\id_{R^{n}}\otimes\mathcal{M}_{\Lambda^{\otimes n}})\left(  \psi_{R^{n}%
A^{n}}^{\rho} \right)  \right\}  +d_{\max}\varepsilon\\
&  \leq D+d_{\max}\varepsilon,
\end{align*}
where $d_{\max}$ is the maximum eigenvalue of $\Delta^{\left(  n\right)  }$.
Also, in the above, we see how it is possible to derandomize the common
randomness:\ there exists a choice of the $m$ such that%
\begin{equation}
\text{Tr}\left\{  \Delta^{\left(  n\right)  }\sum_{l}\text{Tr}_{A^{n}}\left\{
\id_{R^{n}} \otimes\Lambda_{l}^{\left(  m\right)  }\psi_{R^{n}A^{n}}^{\rho}
\right\}  \otimes\left\vert \mathcal{D}^{\left(  m\right)  }\left(  l\right)
\right\rangle \left\langle \mathcal{D}^{\left(  m\right)  }\left(  l\right)
\right\vert \right\}  \leq D+d_{\max}\varepsilon. \label{eq:derandomize-QC-RD}%
\end{equation}
Hence,%
\[
\lim_{n\rightarrow\infty}{\overline{d}}(\rho,\mathcal{D}_{n}\circ
\mathcal{M}_{\Lambda^{(n)}})\leq D.
\]
Thus, a measurement compression protocol directly yields a q-c rate distortion protocol.

Now we give a proof for the converse. Let $\Lambda^{(n)}:A^{n}\mapsto L$ be a
POVM with $\Lambda^{(n)}=\{\Lambda_{l}^{(n)}\}$, and let $\mathcal{D}%
_{n}:L\mapsto X^{n}$ be a decoding map (with $L$ and $X^{n}$ denoting
classical systems) such that
\begin{equation}
\lim_{n\rightarrow\infty}{\overline{d}}(\rho,\mathcal{D}_{n}\circ
\mathcal{M}_{\Lambda^{(n)}})\leq D, \label{eq:distortion-converse}%
\end{equation}
where $\mathcal{M}_{\Lambda^{(n)}}$ is the measurement map corresponding to
the POVM $\Lambda^{(n)}$. Defining $\sigma_{R^{n}L}\equiv\left(  \id_{R^{n}%
}\otimes\mathcal{M}_{\Lambda^{(n)}}\right)  \psi_{R^{n}A^{n}}^{\rho},$ we have
$\sigma_{L}=\sum_{l}\text{Tr}(\Lambda_{l}^{(n)}\rho^{\otimes n})|l\rangle
\langle l|$ and
\begin{align}
nR  &  \geq H(L)_{\sigma}\nonumber\\
&  \geq I(L;R^{n})_{\sigma}\nonumber\\
&  \geq I(X^{n};R^{n})_{\omega}. \label{last}%
\end{align}
The first inequality holds because the entropy $H(L)_{\sigma}$ is upper
bounded by the entropy $nR$ of the uniform distribution. In the second line,
the inequality follows because $I(L;R^{n})_{\sigma}=H(L)_{\sigma}%
-H(L|R^{n})_{\sigma}$ and $H(L|R^{n})_{\sigma}\geq0$ since $L$ is classical.
In the third line, $\omega_{X^{n}R^{n}D}\equiv(\id_{R^{n}}\otimes
\mathcal{D}_{n})\sigma_{R^{n}L}$. This inequality follows from the quantum
data processing inequality (Lemma~\ref{dataproc}). Continuing, we have%
\begin{align}
{\hbox{RHS of \reff{last} \,}}  &  \geq\sum_{i=1}^{n}I(X_{i};R_{i})\nonumber\\
&  \geq\sum_{i=1}^{n}R^{qc}\left(  d(\rho,\mathcal{F}_{n}^{(i)})\right)
\nonumber\\
&  =n\sum_{i=1}^{n}\frac{1}{n}R^{qc}\left(  d(\rho,\mathcal{F}_{n}%
^{(i)})\right) \nonumber\\
&  \geq nR^{qc}\left(  \sum_{i=1}^{n}\frac{1}{n}d(\rho,\mathcal{F}_{n}%
^{(i)})\right) \nonumber\\
&  \geq nR^{qc}(D),
\end{align}
for $n$ sufficiently large. In the above, $\mathcal{F}_{n}^{(i)}$ is the
marginal operation on the $i$-th copy of the source space induced by the
overall operation $\mathcal{D}_{n}\circ\mathcal{M}_{\Lambda^{(n)}}$. The first
inequality follows from the superadditivity of the quantum mutual information
(Lemma \ref{super}). The second inequality follows from the fact that the map
$\mathcal{F}_{n}^{(i)}$ has distortion $d(\rho,\mathcal{F}_{n}^{(i)})$, which,
by definition (\ref{def1}), is lower bounded by the q-c rate distortion
function corresponding to this distortion. The last two inequalities follow
from the convexity of the q-c rate distortion function, from the assumption
that the average distortion of the protocol is less than or equal to $D$ for
$n$ large enough, i.e.,
\[
\sum_{i=1}^{n}\frac{1}{n}d(\rho,\mathcal{F}_{n}^{(i)})\leq D,\quad
{\hbox{for $n$ sufficiently large}},
\]
and the fact that $R^{qc}(D)$ is a non-increasing function of $D$.
\end{proof}

A natural choice for each $\Delta_{x}$ in the distortion observable in
(\ref{qc}) is%
\begin{equation}
\Delta_{x}=I-|x\rangle\langle x|.
\end{equation}
For such a choice, the distortion of the classical data, resulting from
the measurement, is measured with respect to the classical data that would result from
an ideal measurement of the source state $\rho$ in its eigenbasis. However, such a choice is effectively classical because the operators $\Delta_{x}$
are diagonal in the Schmidt basis of $\psi_{RA}^{\rho}$. We show in
Lemma~\ref{lem:no-better-than-shannon}\ below that, for such a choice of the distortion observable, the best strategy for rate-distortion
coding amounts to an effectively classical strategy, in which Alice measures each
output of the source state in its eigenbasis, thus obtaining a classical sequence, which she then compresses by applying the purely classical protocol for Shannon's rate-distortion coding. Thus, a necessary condition for there to be a quantum
advantage in quantum-to-classical rate distortion coding is that the operators
$\Delta_{x}$ should not be diagonal in the Schmidt basis of $\psi_{RA}^{\rho}%
$. After Lemma~\ref{lem:no-better-than-shannon}, we provide an example of a
quantum source and a distortion observable for which quantum-to-classical rate
distortion coding gives an advantage over the above classical strategy.

\begin{lemma}
\label{lem:no-better-than-shannon}If each operator $\Delta_{x}$, in the definition \reff{qc} of the distortion observable, is diagonal in
the Schmidt basis of $\psi_{RA}^{\rho}$, then 
a quantum-to-classical rate distortion coding scheme has no advantage over a classical scheme, in the following sense: the optimal measurement map is a 
von Neumann measurement in the eigenbasis of the source state, followed
by  classical
post-processing of the measurement result according to Shannon's rate
distortion theory.
\end{lemma}

\begin{proof}
Let $\Lambda$\ denote the minimal POVM\ in (\ref{up1}) for a
given distortion $D$, and let $\cM_{\Lambda}$ denote the corresponding
measurement map. Let the Schmidt decomposition of the purification 
$\psi_{RA}^{\rho}$ of the source state $\rho\in \cD(\cH_A)$ be
as follows:%
\[
\left\vert \psi^{\rho}_{RA}\right\rangle =\sum_{z}\sqrt{p\left(  z\right)
 }\left\vert z\right\rangle _{R}\left\vert z\right\rangle _{A},
 \]
where $p(z)$ are the Schmidt coefficients.
Then the distortion that the map $\cM_{\Lambda}$ causes is as follows:%
\begin{align*}
&  \text{Tr}\left\{  \Delta_{RX}\left(  \text{id}_{R}\otimes\mathcal{M}%
_{\Lambda}\right)  \left(  \psi_{RA}^{\rho}\right)  \right\} \\
&  =\text{Tr}\left\{  \left(  \sum_{x}\Delta_{x}\otimes\left\vert
x\right\rangle \left\langle x\right\vert _{X}\right)  \left(  \sum
_{z,z^{\prime},y}\sqrt{p\left(  z\right)  p\left(  z^{\prime}\right)
}\left\vert z\right\rangle \left\langle z^{\prime}\right\vert _{R}%
\otimes\text{Tr}\left\{  \Lambda_{y}\left\vert z\right\rangle \left\langle
z^{\prime}\right\vert _{A}\right\}  \left\vert y\right\rangle \left\langle
y\right\vert \right)  \right\} \\
&  =\sum_{x,z,z^{\prime}}\sqrt{p\left(  z\right)  p\left(  z^{\prime
}\right)  }\left\langle z^{\prime}\right\vert \Delta_{x}\left\vert
z\right\rangle \ \left\langle z^{\prime}\right\vert \Lambda_{x}\left\vert
z\right\rangle _{A},
\end{align*}
which is equivalent to%
\begin{equation}
\text{Tr}\left\{  \left(  \sum_{x}\Delta_{x}\otimes\Lambda_{x}\right)  \left(
\psi_{RA}^{\rho}\right)  \right\}  . \label{eq:alt-dist-exp}%
\end{equation}
Now suppose 
that each $\Delta_{x}$ is diagonal in
the Schmidt basis $\{|z\rangle_R\}$ of the reference system, so that%
\[
\Delta(x)=\sum_{z}\left\langle z\right\vert \Delta
_{x}\left\vert z\right\rangle_R\left\vert z\right\rangle \left\langle z\right\vert_R.
\]
Then the above expression for the distortion reduces to the following one:%
\be\label{equiv}
\sum_{x,z}p\left(  z\right)  \left\langle z\right\vert \Delta
_{x}\left\vert z\right\rangle_R \ \left\langle z\right\vert \Lambda
_{x}\left\vert z\right\rangle _{A}.
\ee
Consider $\cH_R \simeq \cH_A$ and choose $\{|z\rangle_R\}$ and  $\{|z\rangle_A\}$ to be identical
bases, which we simply denote as $\{|z\rangle\}$. Then \reff{equiv} implies that,
starting from the original POVM\ $\Lambda$, 
we can construct another POVM 
$\Lambda^{\prime}$\ (say) which is diagonal in the eigenbasis of $\rho$, and 
which results in a distortion equal to that caused by the original POVM. 
The POVM $\Lambda^{\prime}$ is given by
\[
\Lambda^{\prime}:=\left\{  \Lambda_{x}^{\prime}\right\} , \quad {\hbox{where}} \quad \Lambda_{x}^{\prime}:=\sum_{z}
\left\langle z\right\vert \Lambda_{x}\left\vert z\right\rangle
\left\vert
z\right\rangle \left\langle z\right\vert.
\]
Clearly, the following identity holds%
\[
\sum_{x,z}p\left(  z\right)  \left\langle z\right\vert \Delta
_{x}\left\vert z\right\rangle \ \left\langle z\right\vert \Lambda
_{x}\left\vert z\right\rangle=\sum_{x,z}p\left(  z\right)
\left\langle z\right\vert \Delta_{x}\left\vert z\right\rangle \ \left\langle
z\right\vert \Lambda_{x}^{\prime}\left\vert z\right\rangle.
\]
The joint state of the reference system and the post-measurement
classical register, resulting from the POVMs $\Lambda$ and $\Lambda^\prime$, are respectively given as follows:%
\be
\sigma_{RX}= (\id_{R}\otimes\mathcal{M}_{\Lambda})(\psi_{RA}^{\rho}), \quad {\hbox{and}} \quad \sigma_{RX}^\prime= (\id_{R}\otimes\mathcal{M}_{\Lambda^{\prime}})(\psi_{RA}^{\rho}),
\ee
where $\mathcal{M}_{\Lambda^{\prime}}$ is the measurement map corresponding to the POVM $\Lambda^{\prime}$.
It turns out that the mutual information $I(X;R)_\sigma$ 
can only be smaller than $I(X;R)_{\sigma^{\prime}}$.
This can be seen as follows. Note that we can equivalently 
write the state 
$\sigma_{RX}$ as
\begin{equation}
\sigma_{RX}=\sum_{x}\left(\sqrt{\rho}\Lambda_{x}^{T}\sqrt{\rho}\right)_{R}\otimes\left\vert
x\right\rangle \left\langle x\right\vert _{X}. \label{eq:cq-state-for-qc-2}%
\end{equation}
Then the state $\sigma^{\prime}_{RX}$ can be written as%
\[
\sigma^{\prime}_{RX}=\sum_{x}\left(\sqrt{\rho}\Bigl(  \sum_{z}\left\langle
z\right\vert \Lambda_{x}^{T}\left\vert z\right\rangle\left\vert z\right\rangle  \left\langle
z\right\vert \Bigr)  \sqrt{\rho}\right)_{R}\otimes\left\vert x\right\rangle
\left\langle x\right\vert _{X}.
\]
Since $\{|z\rangle\}$ is the eigenbasis of $\rho$, it follows
that $[\sqrt{\rho}, |z\rangle \langle z|]=0$,
and hence the above state is equivalent to the following one:%
\[
\sum_{x}\left(\sum_{z} \left\langle z\right\vert \left(
\sqrt{\rho}\Lambda_{x}^{T}\sqrt{\rho}\right)  \left\vert z\right\rangle
\left\vert z\right\rangle \left\langle z\right\vert \right)\otimes\left\vert x\right\rangle \left\langle
x\right\vert _{X},
\]
which is a classical-classical state. 
Note that such a state is equivalent to the state which
would result from the action of a completely dephasing channel on the reference system $R$ of the state $\sigma_{RX}$ given by 
(\ref{eq:cq-state-for-qc-2}), i.e., $\sigma_{RX}^{\prime}= (\cN \otimes {\rm{id}})\sigma_{RX}$, where $\cN$ denotes a completely dephasing 
channel.
The mutual information can only decrease under such a map and hence
$I(X;R)_{\sigma^\prime} \le  I(X;R)_\sigma$.
This implies that in this case the optimal measurement to perform on the source is a von Neumann measurement in the eigenbasis of $\rho$, followed by classical
post-processing according to the conditional distribution given by $p\left(
x|z\right)  \equiv\left\langle z\right\vert \Lambda_{x}\left\vert
z\right\rangle $ (that this is a distribution follows from the fact that
$\sum_{x}\Lambda_{x}=I$). Thus, this is equivalent to what one would obtain by
exploiting Shannon's rate distortion theorem in a straightforward way.
\end{proof}
\begin{figure}
[ptb]
\begin{center}
\includegraphics[
width=3.5405in
]%
{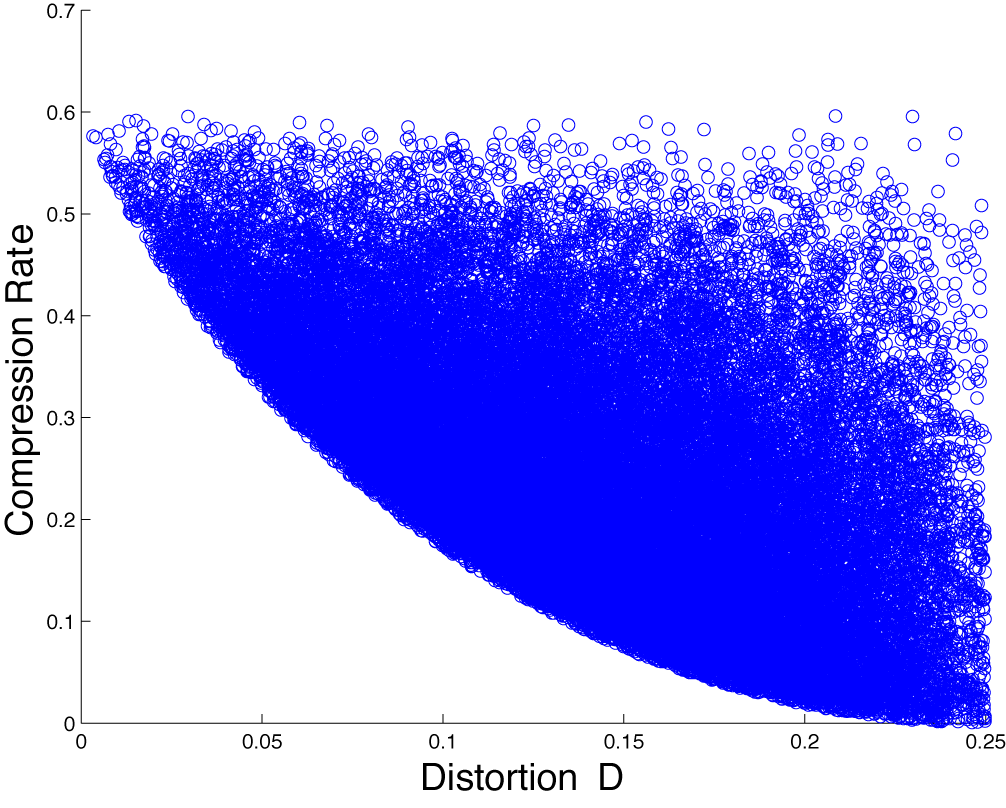}%
\caption{A plot of compression rate vs.~distortion for the quantum information source $\rho$ given by \reff{source_state}, and the rate distortion observable given by \reff{dist_obs}. It was obtained by randomly sampling 250,000 two-outcome POVMs, 
and (for those POVMs which satisfy the distortion criterion $D\le 1/4$) 
plotting the mutual information $I(X;R)_\sigma$ for the resulting state $\sigma_{RX}$ (defined by \reff{eq:cq-state-for-qc}) against the corresponding value of the distortion. The boundary of the shaded region defines the rate-distortion trade-off curve.}%
\label{fig:qc-curve}%
\end{center}
\end{figure}
\bigskip

The following example illustrates a scenario in which a quatum-to-classical rate distortion coding gives an advantage over a purely classical strategy.
\medskip

\noindent
{\bf{Example:}} Consider a quantum information source which 
generates the states $\left\vert
+\right\rangle $ and $\left\vert 0\right\rangle $ with equal 
probability 1/2, so that the density operator for the source is%
\be\label{source_state}
\rho=1/2\left(  \left\vert +\right\rangle \left\langle +\right\vert +\left\vert
0\right\rangle \left\langle 0\right\vert \right)  =\cos^{2}\left(
\pi/8\right)  \left\vert \phi_{0}\right\rangle \left\langle \phi
_{0}\right\vert +\sin^{2}\left(  \pi/8\right)  \left\vert \phi_{1}%
\right\rangle \left\langle \phi_{1}\right\vert ,
\ee
where%
\begin{align*}
\left\vert \phi_{0}\right\rangle  & \equiv\cos\left(  \pi/8\right)  \left\vert
0\right\rangle +\sin\left(  \pi/8\right)  \left\vert 1\right\rangle ,\\
\left\vert \phi_{1}\right\rangle  & \equiv\sin\left(  \pi/8\right)  \left\vert
0\right\rangle -\cos\left(  \pi/8\right)  \left\vert 1\right\rangle .
\end{align*}
A purification of the source state $\rho$ is given by%
\[
|\psi^\rho_{RA}\rangle = \cos\left(  \pi/8\right)  \left\vert \phi_{0}\right\rangle _{R}\left\vert
\phi_{0}\right\rangle _{A}+\sin\left(  \pi/8\right)  \left\vert \phi
_{1}\right\rangle _{R}\left\vert \phi_{1}\right\rangle _{A}.
\]
Suppose we are interested in measuring the distortion of the classical data 
(obtained as a result of a quantum-to-classical rate distortion task), by how much it deviates from the
quantum states that specify the ensemble of the quantum information source. In this case, we would
choose our distortion observable to be as follows:%
\be\label{dist_obs} 
\Delta_{RX}  =\left(  I-\left\vert +\right\rangle \left\langle +\right\vert
\right)  _{R}\otimes\left\vert 0\right\rangle \left\langle 0\right\vert
_{X}+\left(  I-\left\vert 0\right\rangle \left\langle 0\right\vert \right)
_{R}\otimes\left\vert 1\right\rangle \left\langle 1\right\vert _{X}.
\ee

Note that if we consider a two-outcome POVM $\Lambda=\{\Lambda_0, \Lambda_1\}$, where $\Lambda_0 = {I}/{2} = \Lambda_1$, then the state $\sigma_{RX}$ defined 
by \reff{eq:cq-state-for-qc} is given by 
$$
\sigma_{RX} = \rho_R \otimes \frac{I}{2},
$$
where$$\rho_R = \tr_A \{ \psi^\rho_{RA} \} = \cos^2\left(  \pi/8\right)  \left\vert \phi_{0}\right\rangle
\left\langle \phi_{0}\right\vert_{R} +\sin^2\left(  \pi/8\right)  \left\vert \phi
_{1}\right\rangle \left\langle \phi_{1}\right\vert_{R}.$$
In this case, the choice \reff{dist_obs} of the distortion observable yields the following value of the distortion:%
$$D\equiv \text{Tr}\left\{  \Delta_{RX}\left(  \text{id}_{R}\otimes\mathcal{M}%
_{\Lambda}\right)  \left(  \psi_{RA}^{\rho}\right)  \right\} = 1/4.
$$
Moreover, since the state $\sigma_{RX}$ is uncorrelated, we have that $I(X;R)_\sigma =0$, and hence, by Theorem~\ref{thm:qc-rd}, the rate distortion function $R^{qc}(D)$ is equal to zero. This implies that to obtain the full rate-distortion trade-off curve, one only needs to consider values of the distortion $D$ in the range $0\le D\le 1/4$.

The rate-distortion trade-off curve, for the above range of values of $D$, was obtained numerically for the rate distortion
observable defined by \reff{dist_obs}, and is given by the boundary of the shaded region in Fig.~\ref{fig:qc-curve}. As expected, the curve decreases monotonically with $D$.

To prove that in this case a quantum-to-classical rate distortion coding gives an advantage over a purely classical strategy, consider a two-outcome POVM $\Lambda$ which corresponds to a von Neumann measurement in the eigenbasis of the source state $\rho$, i.e., $\Lambda = \{ \Lambda_0, \Lambda_1\}$, where
$$ \Lambda_0 = |\phi_0\rangle \langle \phi_0| \quad {\hbox{and}} \quad  \Lambda_1 = |\phi_1\rangle \langle \phi_1|,$$
or, more generally, consider any  $\Lambda = \{ \Lambda_0, \Lambda_1\}$ such that 
$\cN(\Lambda_i) = \Lambda_i$ for $i=0,1$, where $\cN$ denotes a dephasing channel, with the dephasing being in the eigenbasis of $\rho$.
In this case one finds that, if the distortion observable is 
chosen as in \reff{dist_obs}, the distortion is always equal to the maximum allowed value $D=1/4$. This implies that for distortion in the range $0\le D < 1/4$, for the choice \reff{dist_obs}, quantum-to-classical rate-distortion coding gives an advantage over a classical strategy.\footnote{A ``classical strategy'' here corresponds to a measurement in the eigenbasis of $\rho$, followed by classical post-processing.}

\section{Quantum-to-classical rate-distortion coding with quantum side
information}

\label{sec-qsi}We now consider a class of protocols in which Alice and Bob
share many copies of some quantum state $\rho_{AB}$. This state can be
considered to arise from the action of an isometry on the state of a
memoryless quantum information source performed by a third party (say,
Charlie), who then distributes the systems $A$ and $B$ to Alice and Bob,
respectively. The system $B$ acts as Bob's quantum side information. We also
let Alice and Bob share common randomness. The goal is to quantify the minimum
rate at which Alice needs to send classical data to Bob, such that he can
reconstruct a classical approximation of the state $\rho_{A}=$ Tr$_{B}\left\{
\rho_{AB}\right\}  $ by using the received classical data and his quantum side
information. By a \textquotedblleft classical approximation,\textquotedblright%
\ we mean that for a fixed distortion $D\geq0$, where the distortion is
defined as%
\begin{equation}
d\left(  \rho,\mathcal{M}_{\Lambda}\right)  \equiv\text{Tr}\left\{
\Delta_{RXB}\left(  \text{id}_{R}\otimes\mathcal{M}_{\Lambda}\otimes
\text{id}_{B}\right)  (\psi_{RAB}^{\rho})\right\}  , \label{eq:dist-meas-qsi}%
\end{equation}
and a chosen distortion observable $\Delta$ of the following form:%
\begin{equation}
\Delta_{RBX}\equiv\sum_{x}\Delta_{RB}^{x}\otimes\left\vert x\right\rangle
\left\langle x\right\vert _{X}, \label{eq:qc-qsi-dist-obs}%
\end{equation}
we require that (\ref{eq:rate-dist-criterion}) is satisfied. The rate
distortion function in this scenario is defined in a manner analogous to
$R^{qc}(D)$ of the previous section, and is denoted as $R_{qsi}^{qc}\left(
D\right)  $. In the above, $\psi_{RAB}^{\rho}$ is a purification of the state
$\rho_{AB}$, and since we are interested in measuring the distortion that
occurs on the $A$ system only, the operators $\Delta_{RB}^{x}$ in
(\ref{eq:qc-qsi-dist-obs}) should act on all systems that purify the $A$
system. Figure~\ref{fig:qc-rd-QSI}\ depicts the most general protocol for
quantum-to-classical rate-distortion coding with quantum side information.%
\begin{figure}
[ptb]
\begin{center}
\includegraphics[
natheight=3.739500in,
natwidth=6.653900in,
height=2.0029in,
width=3.5397in
]%
{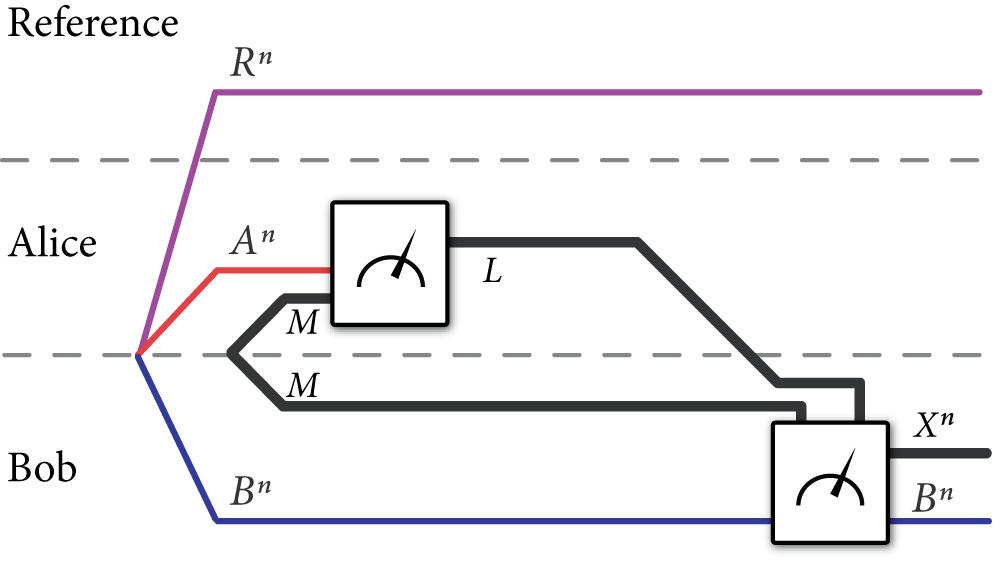}%
\caption{The most general protocol for quantum-to-classical rate-distortion
coding with quantum side information. Alice and Bob share many copies of a
quantum state $\rho_{AB}$, which is purified by an inaccessible reference
system. We also allow them access to common randomness $M$\ before the
protocol begins. Alice first performs a collective measurement on her systems,
producing a classical output $L$. She then transmits $L$ over noiseless
classical bit channels to Bob. Bob performs a collective measurement on his
quantum systems, depending on what he receives from Alice and his share of the
common randomness. This measurement produces a classical sequence $X^{n}$ and
has quantum outputs as well. The protocol is deemed successful if the
classical sequence $X^{n}$\ is not distorted on average from the quantum
source more than a specified amount according to a suitable distortion
observable. We also demand that the disturbance caused by the protocol to the
joint state of the reference and Bob's systems is asymptotically negligible.
This in turn implies that quantum side information suffers a negligible
disturbance and hence is available to Bob for future use.}%
\label{fig:qc-rd-QSI}%
\end{center}
\end{figure}

Ref.~\cite{WHBH12} contains a theorem that determines the optimal rates for
measurement compression in the presence of quantum side information. It almost
immediately leads to the following rate distortion theorem:

\begin{theorem}
\label{thm:qc-qsi} For a memoryless quantum information source characterized
by a state $\rho_{AB}$ (where Alice possesses $A$ and Bob possesses $B$), a quantum-to-classical
distortion observable $\Delta_{RBX}$, and
any given distortion $D\geq0$, an achievable rate for quantum-to-classical
rate distortion with quantum side information, when sufficient common
randomness is available, is given by%
\begin{equation}
\min_{\Lambda\ :\ d\left(  \rho,\mathcal{M}_{\Lambda}\right)  \leq D}I\left(
X;R|B\right)  _{\sigma}, \label{eq:qc-qsi-rd}%
\end{equation}
so that%
\begin{equation}
R_{qsi}^{qc}\left(  D\right)  \leq\min_{\Lambda\ :\ d\left(  \rho
,\mathcal{M}_{\Lambda}\right)  \leq D}I\left(  X;R|B\right)  _{\sigma},
\label{25a}%
\end{equation}
where $\Lambda\equiv\{\Lambda_{x}\}$ is a POVM acting only on Alice's system,
$d\left(  \rho,\mathcal{M}_{\Lambda}\right)  $ is defined through
(\ref{eq:dist-meas-qsi})-(\ref{eq:qc-qsi-dist-obs}), and $\psi_{RAB}^{\rho}$
is a purification of the state $\rho_{AB}$. The state $\sigma$ is the
following classical-quantum state:%
\begin{equation}
\sigma_{XRB}\equiv\sum_{x}\left\vert x\right\rangle \left\langle x\right\vert
_{X}\otimes\emph{Tr}_{A}\left\{  \left(  I_{R}\otimes\Lambda_{x}\otimes
I_{B}\right)  \left(  \psi_{RAB}^{\rho}\right)  \right\}  . \label{sigma}%
\end{equation}

\end{theorem}

\begin{proof}
The proof of the achievability part of this theorem proceeds similarly to that
of Theorem~\ref{thm:qc-rd}. We merely fix the POVM\ that minimizes the RHS\ of
(\ref{eq:qc-qsi-rd}). From this POVM, we can construct a protocol for
measurement compression with quantum side information by invoking Theorem~12
of ~\cite{WHBH12}. This protocol exploits classical communication at a rate
$I\left(  X;R|B\right)  _{\sigma}$ and common randomness at a rate $H\left(
X|RB\right)  _{\sigma}$ in order to simulate the action of the POVM\ on many
copies of the state $\rho_{AB}$. By an argument similar to that in the proof
of the achievability part of Theorem~\ref{thm:qc-rd}, we know that such a
protocol meets the distortion criterion and that it is possible to derandomize
the common randomness in the same way as in (\ref{eq:derandomize-QC-RD}).
\end{proof}

If, in addition, we demand that the protocol causes asymptotically negligible
disturbance of the state of Bob (i.e., the quantum side information) and the
state of the reference system, then we can prove that the upper bound in
(\ref{25a}) is achieved. Hence, in this case, the rate distortion function,
which we denote as $\tR^{qc}_{qsi}$, is
given by a single-letter formula. The requirement of the protocol leaving the
states of Bob and the reference essentially undisturbed might seem somewhat
restrictive at first. However, it can be justified as follows. Firstly, note
that ignoring the quantum side information leads to a protocol with a
classical communication rate of $I\left(  X;RB\right)  $ which of course does
not disturb the systems of the reference and Bob in any way. Secondly, Bob
might wish to use the quantum side information in some future
information-processing task, which therefore leads to the above requirement on
the state of his system. In light of this, it seems reasonable to restrict
consideration to a class of protocols in which Bob is allowed to exploit the
quantum side information, but only in a way which causes negligible
disturbance to it. These considerations yield the following theorem:

\begin{theorem}
\label{thm:qc-cr-qsi} For a memoryless quantum information source $\rho_{AB}$
(where Alice possesses $A$ and Bob possesses $B$), a quantum-to-classical
distortion observable $\Delta_{RBX}$, and any given distortion
$D\geq0$, the quantum-to-classical rate distortion function with quantum side
information, sufficient common randomness, and such that the protocol causes
only a negligible disturbance to the systems of the reference and Bob, is
given by%
\begin{equation}
\tR^{qc}_{qsi}\left(  D\right)  =\min_{\Lambda\ :\ d\left(  \rho,\mathcal{M}%
_{\Lambda}\right)  \leq D}I\left(  X;R|B\right)  _{\sigma},
\label{eq:single-letter-QC-QSI-RD}%
\end{equation}
where the state $\sigma$ is as defined in (\ref{sigma}) of
Theorem~\ref{thm:qc-qsi}.
\end{theorem}

\begin{proof}
The proof of the achievability part of this theorem again follows directly
from Theorem~12 of Ref.~\cite{WHBH12} which deals with measurement compression
in the presence of quantum side information. We merely fix the map that
minimizes the expression in (\ref{eq:single-letter-QC-QSI-RD}) and apply the
aforementioned theorem. The resulting protocol meets the distortion constraint
because of the way that the POVM\ is chosen in
(\ref{eq:single-letter-QC-QSI-RD}).

The converse part of this theorem exploits the approach from the converse
parts of Theorems~12 and 14 of Ref.~\cite{WHBH12}, which in turn exploit ideas
of Cuff \cite{C08}. The most general protocol begins with the state $\left(
\psi^{\rho}_{RAB}\right)  ^{\otimes n}$ shared between the reference, Alice,
and Bob. We let Alice and Bob share common randomness as well (embodied in
some random variable $M$). Alice performs an encoding on her systems $A^{n}$
with the help of her share of the common randomness $M$, producing a classical
output given by the random variable $L$ which takes values in a finite
alphabet $\mathcal{L}$. Let $\sigma$ denote the state at this point. Also, let
$R$ be the rate of classical communication, i.e., $R = (\log_{2}
|\mathcal{L}|)/n$. Alice sends $L$ to Bob, who then combines this with his
share of the common randomness to perform some decoding map on $B^{n}$,
producing a classical sequence $X^{n}$ and a quantum system $B^{\prime n}$.
Let $\omega$ denote the final state after this encoding-decoding procedure.
Further, let ${\mathcal{F}}_{n}^{A^{n}B^{n} \to X^{n}B^{\prime n}}$ denote the
effective CPTP map on $\rho_{AB}^{\otimes n}$ (the state that Alice and Bob
share at the start of the protocol) resulting from these encoding and decoding
operations.
We demand that the distortion of the output $X^{n}$ be no larger than $D$ (in
a sense similar to that in (\ref{eq:distortion-converse}), though in this case
we need to trace over the systems $B^{\prime n}$), and we furthermore demand
that the trace distance between $\left(  \psi^{\rho}_{RB}\right)  ^{\otimes
n}$ and the state $\omega_{R^{n}B^{\prime n}}$\ on systems $R^{n}B^{\prime n}$
(at the end of the protocol) be no larger than some arbitrarily small
$\varepsilon>0$. The converse then proceeds as follows. For $n$ large enough,
\begin{align}
\label{38}nR  &  \geq H\left(  L\right)  _{\sigma}\nonumber\\
&  \geq I\left(  L;MB^{n}R^{n}\right)  _{\sigma}\nonumber\\
&  =I\left(  LMB^{n};R^{n}\right)  _{\sigma}+I\left(  L;MB^{n}\right)
_{\sigma}-I\left(  R^{n};B^{n}M\right)  _{\sigma}\nonumber\\
&  \geq I\left(  LMB^{n};R^{n}\right)  _{\sigma}-I\left(  R^{n};B^{n}\right)
_{\sigma}\nonumber\\
&  \geq I\left(  X^{n}B^{\prime n};R^{n}\right)  _{\omega}-I\left(
R^{n};B^{\prime n}\right)  _{\omega}-n\varepsilon^{\prime}\nonumber\\
&  \geq\sum_{k}\left[  I\left(  X_{k}B_{k}^{\prime};R_{k}\right)  _{\omega
}-I\left(  R_{k};B_{k}^{\prime}\right)  \right]  -2n\varepsilon^{\prime
}\nonumber\\
&  =\sum_{k}I\left(  X_{k};R_{k}|B_{k}^{\prime}\right)  _{\omega
}-2n\varepsilon^{\prime}.%
\end{align}
The first inequality follows because the entropy of a system is always less
than the logarithm of its dimension. The second inequality follows because
$I\left(  L;MB^{n}R^{n}\right)  _{\sigma}=H\left(  L\right)  _{\sigma
}-H\left(  L|MB^{n}R^{n}\right)  _{\sigma}$ and $H\left(  L|MB^{n}%
R^{n}\right)  _{\sigma}\geq0$ for a classical $L$. The first equality is an
identity for quantum mutual information. The third inequality follows because
the common randomness $M$ is in a product state with $R^{n}B^{n}$ so that
$I\left(  R^{n};B^{n}M\right)  _{\sigma}=I\left(  R^{n};B^{n}\right)
_{\sigma}$ and because $I\left(  L;MB^{n}\right)  _{\sigma}\geq0$. The fourth
inequality follows from quantum data processing, Lemma~\ref{dataproc}, (the
systems $LMB^{n}$ are processed to produce systems $X^{n}B^{\prime n}$), and
from the requirement that the protocol causes negligible disturbance of the
state of $R^{n}B^{n}$. The term $\varepsilon^{\prime}$ (which is a function of
${\varepsilon}$) arises from an application of the Alicki-Fannes' inequality
\cite{AF04}, where $\lim_{\varepsilon\rightarrow0}\varepsilon^{\prime}\left(
\varepsilon\right)  =0$. The fifth inequality follows from superadditivity of
quantum mutual information (Lemma~\ref{super}) and because the state on
$R^{n}B^{\prime n}$ is close in trace distance to a tensor-product state (see
Lemma~10 of Ref.~\cite{WHBH12}). The second equality follows from the identity
$I\left(  X_{k}B_{k}^{\prime};R_{k}\right)  _{\omega}-I\left(  R_{k}%
;B_{k}^{\prime}\right)  =I\left(  X_{k};R_{k}|B_{k}^{\prime}\right)  _{\omega
}$.

At this point, we have argued that the above lower bound holds for a protocol
that exploits common randomness and classical communication to implement a map
$\mathcal{F}_{n}^{A^{n}B^{n}\rightarrow X^{n}B^{\prime n}}$. This map meets
the distortion constraint while also causing only a negligible disturbance to
the state on $R^{n}B^{n}$, in the sense that%
\[
\left\Vert \text{Tr}_{X^{n}}\left\{  \mathcal{F}_{n}^{A^{n}B^{n}\rightarrow
X^{n}B^{\prime n}}\left(  \left(  \psi^{\rho}_{RAB}\right)  ^{\otimes
n}\right)  \right\}  -\left(  \psi{\rho}_{RB}\right)  ^{\otimes n}\right\Vert
_{1}\leq\varepsilon.
\]
As in the proof of Theorem~14 of Ref.~\cite{WHBH12}, applying Uhlmann's
theorem to the above condition guarantees that there is some
map acting only on Alice's system, such that the information quantity in the
last line of the above chain of inequalities (\ref{38}) does not change too
much. For completeness, we repeat the argument here. Let the Kraus
representation of $\mathcal{F}_{n}^{A^{n}B^{n}\rightarrow X^{n}B^{\prime n}}$
be given by%
\[
\mathcal{F}_{n}^{A^{n}B^{n}\rightarrow X^{n}B^{\prime n}}\left(  \cdot\right)
=\sum_{i}F_{i}\left(  \cdot\right)  F_{i}^{\dag}.
\]
A purification of Tr$_{X^{n}}\left\{  \mathcal{F}_{n}^{A^{n}B^{n}\rightarrow
X^{n}B^{\prime n}}\left(  \left(  \psi^{\rho}_{RAB}\right)  ^{\otimes
n}\right)  \right\}  $ is given by%
\begin{equation}
\sum_{i}F_{i}\left(  \left\vert \psi^{\rho}_{RAB}\right\rangle \right)
^{\otimes n}\otimes\left\vert i\right\rangle _{I}, \label{eq:Uhlmann-ideal}%
\end{equation}
where $I$ is a purifying system, while a purification of $\left(  \psi^{\rho
}_{RB}\right)  ^{\otimes n}$ is $\left(  \left\vert \psi^{\rho}_{RAB}%
\right\rangle \right)  ^{\otimes n}$. By Uhlmann's theorem, there is an
isometry $U^{A^{n}\rightarrow X^{n}I}$\ acting only on Alice's system, taking
$\left(  \left\vert \psi^{\rho}_{RAB}\right\rangle \right)  ^{\otimes n}$ to
an approximation of the state in (\ref{eq:Uhlmann-ideal}) such that the trace
distance between this state and $U^{A^{n}\rightarrow X^{n}I}\left(  \left\vert
\psi^{\rho}_{RAB}\right\rangle \right)  ^{\otimes n}$ is at most
$2\sqrt{\varepsilon}$. Thus, the map on Alice's side consists of applying
$U^{A^{n}\rightarrow X^{n}I}$ and tracing out $I$. Let $\omega^{\prime}$
denote the resulting state. By exploiting this map instead of the original
one, we find the following lower bound on the information quantity in
(\ref{38}):%
\[
\sum_{k}I\left(  X_{k};R_{k}|B_{k}\right)  _{\omega^{\prime}}-3n\varepsilon
^{\prime}.
\]
The important feature of this approximation map is that it acts only on
Alice's side. Continuing, we have%
\begin{align*}
&  \geq\sum_{k}R^{qc}_{qsi}\left(  d\left(  \rho,\mathcal{G}_{n}^{\left(
k\right)  }\right)  \right)  -3n\varepsilon^{\prime}\\
&  =n\sum_{k}\frac{1}{n}R^{qc}_{qsi}\left(  d\left(  \rho,\mathcal{G}%
_{n}^{\left(  k\right)  }\right)  \right)  -3n\varepsilon^{\prime}\\
&  \geq nR^{qc}_{qsi}\left(  \sum_{k}\frac{1}{n}d\left(  \rho,\mathcal{G}%
_{n}^{\left(  k\right)  }\right)  \right)  -3n\varepsilon^{\prime}\\
&  \geq nR^{qc}_{qsi}\left(  D\right)  -3n\varepsilon^{\prime}.%
\end{align*}
In the above, $\mathcal{G}_{n}^{\left(  k\right)  }$ is the marginal operation
on the $k^{\text{th}}$ copy of the source space induced by the overall
encoding and approximation of the decoding guaranteed by Uhlmann's theorem.
The first inequality follows from the fact that the map $\mathcal{G}%
_{n}^{\left(  k\right)  }$ has distortion $d( \rho,\mathcal{G}_{n}^{\left(
k\right)  }) $ and the expression (\ref{eq:single-letter-QC-QSI-RD}) for the
rate-distortion function $R^{qc}_{qsi}$ in Theorem~\ref{thm:qc-cr-qsi}
involves a minimum over all maps on Alice's system with this distortion. The
first equality is obvious. The last two inequalities follow because the
rate-distortion function is convex and non-increasing as a function of $D$
(the proof of convexity is similar to the proof of Lemma~14 of
Ref.~\cite{DHW11}, though here we rely on the map acting solely on Alice's system).
\end{proof}

A special case of the above theorem is the setting considered in Theorem 4.2
of Ref.~\cite{LD09}. There, Luo and Devetak considered the scenario in which
the source is a classical-quantum state of the form:%
\[
\sum_{y}p_{Y}\left(  y\right)  \left\vert y\right\rangle \left\langle
y\right\vert _{Y}\otimes\rho_{B}^{y},
\]
where Alice possesses $Y$ and Bob $B$. The goal is for Alice to transmit her
classical data to Bob up to some distortion, and Bob is allowed to use the
quantum side information to help reduce the communication costs. They proved
that the following rate is achievable:%
\[
\min_{p_{X|Y}\left(  x|y\right)  \ :\ \mathbb{E}\left\{  d\left(  x,y\right)
\right\}  \leq D}I\left(  X;Y|B\right)  _{\sigma},
\]
for some classical distortion measure $d\left(  x,y\right)  $ and where the
information quantity is with respect to a state of the following form:%
\[
\sum_{y}p_{X|Y}\left(  x|y\right)  p_{Y}\left(  y\right)  \left\vert
y\right\rangle \left\langle y\right\vert _{Y}\otimes\left\vert x\right\rangle
\left\langle x\right\vert _{X}\otimes\rho_{B}^{y}.
\]
Luo and Devetak were not able to find a single-letter characterization of the
rate-distortion function, but with our additional assumptions of sufficient
common randomness and a negligible disturbance of the quantum side
information, our theorem reduces to a single-letter characterization for their
setting. In fact, if one chooses the distortion observable in
(\ref{eq:qc-qsi-dist-obs}) so that the operators $\Delta_{RB}^{x}$ are
diagonal in the Schmidt basis of the $RB$\ systems of $\psi_{RAB}^{\rho}$,
then a similar statement as in Lemma~\ref{lem:no-better-than-shannon}%
\ applies. That is, in this case, it is optimal to measure the $A$ system in
the eigenbasis of $\rho_{A}$ and proceed according to the protocol of Luo and
Devetak in Ref.~\cite{LD09}. As stated above, their protocol is optimal if we
demand that it cause only a negligible disturbance to the state of the
reference and Bob.

\section{Conclusions and discussions}

We have derived a single-letter formula for the quantum-to-classical rate
distortion function. The goal in quantum-to-classical rate-distortion coding
is to provide a compressed classical approximation of a quantum source, up to
some specified level of distortion, as determined by a distortion observable.
The formula is expressed as a minimization of a quantum mutual information over
all quantum-to-classical channels that meet the distortion constraint. In
general, our results show that a collective measurement of the quantum source
is required to obtain optimal compression rates. However, if the distortion
observable has a classical form (so that each operator $\Delta_x$ is diagonal
in the Schmidt basis), then the best strategy for
quantum-to-classical rate-distortion coding ends up being an effectively
classical strategy, in which Alice performs individual measurements of each copy
of the source in its eigenbasis, and processes the resulting classical data
according to Shannon's classical rate-distortion protocol.

We have also derived a single-letter formula for the quantum-to-classical
rate distortion function when the receiver has some 
quantum side information about the source. Our
assumptions are that Alice and Bob share sufficient common
randomness, and that the protocol causes only a negligible disturbance to the
joint state of the reference and the quantum side information. We consider this
latter assumption to be rather natural, since Bob might wish to make use of
his quantum side information in some future protocol. Our results suggest that
it might generally be possible for quantum information-theoretic protocols
that employ quantum side information to be simplified by employing this
assumption, due to the restriction that it imposes on the quantum states at
the output of a given protocol. This assumption is purely non-classical, since
it is always possible to copy classical information before processing it in
any way.

There are some interesting open questions to consider going forward from here.
It would be ideal if we could derandomize the common randomness in the
protocol that uses quantum side information, since it would imply that this
extra resource is unnecessary. However, if we did so, the protocol for
measurement compression with quantum side information could end up causing a
non-negligible disturbance to the joint state of the reference and Bob's systems, for
some of the values of the common randomness. Since our approach in the proof
of the achievability part of the coding theorem relies on this protocol, we have not been able to conclude
that the common randomness is unnecessary. However, the common randomness plays only a
passive role in the converse theorem, and this suggests that it might
ultimately be unnecessary. In order to determine if this is the case, one
would have to consider a different protocol in proving the achievability part
of the coding theorem.
\bigskip

We acknowledge Patrick Hayden and Ke Li for useful discussions. MMW\ acknowledges support
from the Centre de Recherches Math\'{e}matiques at the University of Montreal.
MH received support from the Chancellor's postdoctoral research fellowship,
University of Technology Sydney (UTS), and was also partly supported by the
National Natural Science Foundation of China (Grant No.~61179030) and the
Australian Research Council (Grant No.~DP120103776). AW was supported by the Royal Society, the Philip Leverhulme Trust, EC integrated project QAP (contract IST-2005-15848), the STREPs QICS and QCS, and the ERC Advanced Grant ``IRQUAT''.

\bibliographystyle{plain}
\bibliography{Ref}

\end{document}